\def\masterthesis{0}
\def\cryptology{0} % For journal of cryptology
\def\sigconf{0} %ACM SIG conferfence
\def\big{0} % For using a huge font

\def\draft{0}
\def\anon{0}
\def\shownomenclature{1}
\def\smallbib{1} %This should only be used if not submitted to a journal/conference with proceedings, as these use their own format
\def\toc{0} %Show TOC. Should be set to 1 if paper is long (say, above 30 pages), and for all types of thesis.

\ifnum\big=1
    \documentclass[final,17pt]{extarticle} \usepackage{fullpage}
    
\else
    \ifnum\sigconf=1
        \documentclass[sigconf,final]{acmart}
    \else
        \ifnum\draft=1
           \documentclass[11pt,draft,a4paper]{article}
        \else
            \ifnum\masterthesis=1
                \documentclass[11pt,final,a4paper,titlepage]{article}
            \else
                \documentclass[11pt,final,a4paper]{article}
            \fi
        \fi
        \usepackage{lmodern}
        
    \fi
\fi

\input{macros}

\usepackage{fullpage}
\usepackage[export]{adjustbox}
\usepackage{xspace}
\usepackage{epigraph}
\usepackage{dirtytalk}

\usepackage{csquotes}\DeclareMathAlphabet{\mathpzc}{OT1}{pzc}{m}{it}

\newcommand{\OTbit}{\ensuremath{b}\xspace}
\newcommand{\CoinBit}{\ensuremath{c}\xspace}
\newcommand{\CommBit}{\ensuremath{b}\xspace}

\newcommand{\Amflip}{\ensuremath{\mathsf{AmFlip}}\xspace}

\newcommand{\OT}{\ensuremath{\mathsf{OT}}\xspace}
\newcommand{\abort}{\ensuremath{\mathsf{Abort}}\xspace}

\newcommand{\AmCom}{\ensuremath{\mathsf{AmCom}}\xspace}
\newcommand{\Sup}{\ensuremath{\mathsf{Sup}}\xspace}
\newcommand{\MOEExp}{\ensuremath{\mathsf{(\frac{1}{2}+\frac{1}{2\sqrt{2}})}}\xspace}
\newcommand{\MOEExpInv}{\ensuremath{\mathsf{4-2\sqrt{2}}}\xspace}

\newcommand{\WRadv}{\ensuremath{W_{\Radv^*_1}
}\xspace}

\newcommand{\CRadv}{\ensuremath{ C_{\Radv^*_1}
}\xspace}
\newcommand{\CTMAC}{\ensuremath{\mathsf{CTMAC}}\xspace}
\newcommand{\Bitspace}{\ensuremath{\mathsf{\{0,1\}}}\xspace}

\newcommand{\Cons}{\ensuremath{\mathsf{Cons}}\xspace}

\newcommand{\AmOT}{\ensuremath{\mathsf{AmROT}}\xspace}

\newcommand{\MOE}{\ensuremath{\mathsf{MOE}}\xspace}
\newcommand{\NP}{\ensuremath{\mathsf{NP}}\xspace}
\newcommand{\IP}{\ensuremath{\mathsf{IP}}\xspace}

\newcommand{\Badv}{\ensuremath{\mathsf{\mathcal{B}}}\xspace}
\newcommand{\Cadv}{\ensuremath{\mathsf{\mathcal{C}}}\xspace}
\newcommand{\Radv}{\ensuremath{\mathsf{\mathcal{R}}}\xspace}
\newcommand{\Sadv}{\ensuremath{\mathsf{\mathcal{S}}}\xspace}
\newcommand{\Hash}{\ensuremath{\mathsf{\mathcal{H}}}\xspace}

\newcommand{\Dadv}{\ensuremath{\mathsf{\mathcal{D}}}\xspace}
\newcommand{\Fadv}{\ensuremath{\mathsf{\mathcal{F}}}\xspace}

\newcommand{\Stall}{\ensuremath{\mathsf{Stall}}\xspace}

\newcommand{\Tmac}{\ensuremath{\mathsf{TMAC}}\xspace}

\newcommand{\Stamp}{\ensuremath{\mathsf{\psi}} \xspace}

\begin{document}
\ifnum\masterthesis=0
    \title{
 Quantum Amnesia Leaves Cryptographic Mementos:\\  A  Note On Quantum Skepticism
    %\ifdraft{\\(working draft)}
    }
\fi

\ifnum\anon=0
    \ifnum\masterthesis=0
        \ifnum\sigconf=0
            \ifnum\cryptology=1

                \author{Or Sattath}
                \affil{Department of Computer Science, Ben Gurion University of the Negev, Beersheba, Israel\\
                sattath@post.bgu.ac.il}
                \author{Uriel Shinar}
                \affil{Department of Computer Science, Ben Gurion University of the Negev, Beersheba, Israel\\
                shinaru@post.bgu.ac.il}
            \else
                \author[1]{Or Sattath}
                \author[1]{Uriel Shinar}

                \affil[1]{Computer Science Department, Ben-Gurion University of the Negev}
            \fi
        \else
            \author{Or Sattath}
            \affiliation{%
            \institution{Computer Science Department, Ben-Gurion University of the Negev}
            \country{Israel}}
        \fi
    \fi
\else
    \ifnum\sigconf=0
        \author{}
    \fi
\fi

\ifnum\masterthesis=0
    \ifnum\sigconf=0
        \maketitle
    \fi
\fi

\ifnum\masterthesis=1
    \begin{titlepage}
        \centering
        { Ben-Gurion University of the Negev}
        
        {The Faculty of Natural Sciences}
        
        {\small The Department of Computer Science}
        
        \vspace{2cm}
        
        \includegraphics[scale=\logoscale]{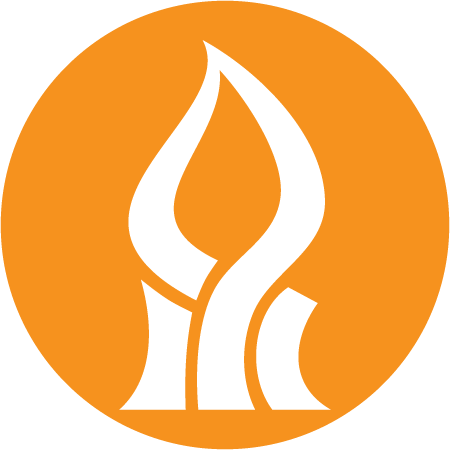}
        
        \vspace{2cm}
        
        {\Large \bfseries Template of Thesis}
        
        \vspace{2cm}
        
        {\small Thesis submitted in partial fulfillment of the requirements for the Master of Sciences degree}
        
        \vspace{1cm}
        
        {\bfseries Name of Student}
        
        {Under the supervision of Dr. Or Sattath}
        
        \vspace{2cm}
        
        \today
    \end{titlepage}
    
    \begin{titlepage}
        \centering
        { Ben-Gurion University of the Negev}
        
        {The Faculty of Natural Sciences}
        
        {\small The Department of Computer Science}
        
        \vspace{2cm}
        
        {\Large \bfseries Template of Thesis}
        
        \vspace{2cm}
        
        {\small Thesis submitted in partial fulfillment of the requirements for the Master of Sciences degree}
        
        \vspace{1cm}
        
        {\bfseries Name of Student}
        
        {Under the supervision of Dr. Or Sattath}
        
        \vspace{1cm}
        
        {\small Signature of student: \longunderline Date: \longunderline}
        
        \vspace{0.5cm}
        
        {\small Signature of supervisor: \longunderline Date: \longunderline}
        
        \vspace{0.5cm}
        
        \begin{changemargin}{-1.5cm}{-1.5cm}
        \centering
            {\small Signature of chairperson of the committee for graduate studies: \longunderline Date: \longunderline}
        \end{changemargin}
        \vspace{2cm}
        
        \today
    \end{titlepage}

    \pagenumbering{roman}
    \begin{center}
        {\large \bfseries Title of Thesis}
        
        \vspace{0.5cm}
        
        {\bfseries Name of Student}
        
        \vspace{0.5cm}
        
        Thesis submitted in partial fulfillment of the requirements for the Master of Sciences degree
        
        \vspace{0.25cm}
        
        {Ben-Gurion University of the Negev}
        
        \vspace{0.25cm}
        
        \today
        
        \vspace{2cm}
        
        {\bfseries Abstract}
        
        \vspace{0.5cm}
    \end{center}
\else 
    \ifnum\sigconf=0
        \begin{abstract}
    \fi
\fi
%IN ARXIV SUBMISSIONS, REMEMBER TO INDENT NEW PARAGRAPHS SINCE NEWLINES ARE IGNORED. 

Leonard Shelby, the protagonist of Memento, uses mementos in the form of tattoos
and pictures to handle his amnesia. Similar to Leonard, contemporary quantum computers
suffer from “quantum amnesia”: the inability to store quantum registers for a long duration.
Quantum computers can only retain classical “mementos” of quantum registers by measuring
them before those vanish. Some quantum skeptics argue that this quantum amnesia is
inherent. We point out that this variant of a skeptic world is roughly described by the quantum bounded storage model, and although it is a computational obstacle that annuls potential quantum computational advantage, the seemingly undesired properties provide a cryptographic advantage. Namely, providing exotic primitives promised by the quantum bounded storage model, such as unconditionally secure commitment and oblivious transfer schemes, with constructions involving nothing but transmission and measurement of BB84 states.

\ifnum\masterthesis=0
    \ifnum\sigconf=0
        \end{abstract}
    \fi
\fi

\ifnum\sigconf=1
    \begin{abstract}
        Abstract of SigConf goes here.
    \end{abstract}
    \maketitle
\fi

\ifnum\masterthesis=1
    \pagebreak
    \pagenumbering{arabic}
    \subsection*{Acknowledgments}
    \setcounter{tocdepth}{3}
    
    \pagebreak
    \tableofcontents
    \listoffigures % Remove if there are no figures
    \listoftables % Remove if there are no tables
    \pagebreak
\fi

\ifnum\toc=1
    \tableofcontents 
\fi

% \section{Introduction} % (fold)
In 1994, Shor presented an efficient quantum algorithm for integer factorization, providing the first evidence that quantum computation might provide superpolynomial speedup over classical computation.
However,  three decades and billions of dollars worth of research funds later, quantum computers have not yet fulfilled their promise. 
 This has caused controversy. Quantum optimists believe we will reach Superiorita: a world where quantum computers have a superpolynomial computational advantage, while the skeptics believe we will forever stay in Skepticland: a world where quantum computers have no such advantage.\footnote{These terms were coined by Barak~\cite{Bar17}, inspired by Impagliazzo's five worlds~\cite{Imp95}.}

One popular objection to quantum advantage is the belief that it is impossible or infeasible to construct long-term quantum memory.
We name the world represented by this belief "Quamnesia": the world where quantum computers suffer from amnesia in the sense that they cannot have long-term quantum memory. 
Quantum optimists agree we currently live in Quamnesia, but believe we will eventually pave a path to Superiorita.
We name those skeptics who think we will forever be trapped in Quamnesia \emph{Quamnesians}.

The fact that we \emph{currently} live in Quamnesia poses a problem for the applicability of quantum algorithms. We offer a  Panglossian\footnote{Dr. Pangloss was the pedantic old tutor in Voltaire's satirical novel Candide. Pangloss was an optimist who claimed that "all is for the best in this best of all possible worlds."} view on this subject in which this perceived difficulty of Quamnesia is actually a blessing in disguise: the inability of parties to store quantum information results in exotic cryptographic primitives known to be impossible even in the quantum standard interactive model.

 Say any quantum state decoheres after $t$ seconds. A party could thus choose to "stall" other parties in a protocol for $t$ seconds before continuing communication. This stalling ensures that nothing but classical mementos are left of any quantum state previously kept by the opposing parties. Of course, the same must apply to the party that chose to stall: following the stall, any quantum state the after-mentioned party has previously held decoheres. 
 
 This aspect of Quamnesia is captured by a special case of the bounded quantum storage model \cite{Sch07, DFSS08}. In the bounded quantum storage model, parties send quantum messages to each other in turns, but a limitation exists on the quantum memory kept by the parties; in some specified stages of the protocol, all but $\gamma$ quantum qubits of every party, are measured in the standard basis. If $\gamma=0$, the model is better described as "no quantum memory," as none at all can be kept---coinciding with the above description of protocols in Quamnesia.

\paragraph{Unconditional Commitment and Oblivious Transfer in Quamnesia}
Quamnesia allows for a unique and exotic bit commitment scheme, first described in \cite{DFSS08}:
In the commit phase, Bob sends a randomly generated tensor product of BB84 qubits(qubits of the set  $\set{\ket{0},\ket{1},\ket{+},\ket{-}}$).
Bob then stalls\footnote{It is convenient to think of the act of stalling as a special message \Stall that the parties can send each other during communication, notifying that no further transmission will be received until $t$ time passes, forcing the parties to measure. We will use this convention in the description of algorithms. } to ensure that the state has decohered, forcing Alice to measure the quantum state sent to her before the decoherence. Alice measures the state she received in the standard basis to commit to $0$, and in the Hadamard basis to commit to $1$. In stark contrast to commitment schemes in the standard model, no data is sent from the committer to the receiver in the commit phase.

In the reveal phase, Alice sends her measurement results and the committed bit. Bob confirms that the measurement results are consistent with the state he sent on the corresponding basis. It is easy to see that the scheme is perfectly hiding since no information is sent to Bob during the commit phase. Ref~\cite[Theorem 7.6, Theorem 7.7]{Sch07} shows a strong result regarding binding in the quantum bounded storage model when $\gamma$ is up to a fourth of the number of qubits sent. We now provide a proof sketch for the case $\gamma=0$, based on the similarity of the commitment to the classically verifiable variant of Wiesner's money.

\begin{theorem}[Binding of $\AmCom$, informal] In Quamnesia, the scheme $\AmCom$ (\cref{alg: Amnesic commitment}) is statistically binding.
\end{theorem}
\begin{proof}[Proof sketch]
In the classically verifiable variant of Wiesner's money \cite{MVW12}, the quantum money states minted to users are BB84 states. Verification is done by sending challenges to the users, requesting the user to measure each qubit in either the standard or Hadamard basis. In the case of a single qubit, opening a commitment to both $0$ and $1$ corresponds to being able to respond to both distinct challenges from the bank on a qubit of the money scheme. A bound on the probability to do so is shown to be $\cos^2{\frac{\pi}{8}}$. For the many qubits case, it can be argued (via semi-definite programming, for example) that the probability of opening a commitment to both $0$ and $1$ decreases exponentially.
Another argument for the binding can be provided by observing an alternate description for the commitment: the scheme effectively consists of the receiver providing the committer with the quantum signature token in the \cite{BSS21} $\Tmac$ (tokenized $\mac$) scheme $\CTMAC$, and the committer then uses that token to sign on the bit it wishes to commit to,  creating a signature that the receiver can verify in the reveal stage, which happens after decoherence.
The security guarantees of $\Tmac$ state that the committer could not have created a signature for both $0$ and $1$ with significant probability, and yet she must have done so to break the binding; the committer can be assumed to have already broken the binding before the reveal stage, due to the promised decoherence. This alternative viewpoint on the commitment scheme has the added benefit of being generic: the same argument would work for any $\Tmac$ scheme satisfying information theoretic unforgeability for a single document.
\end{proof}
Similarly, as a consequence of strong known results in the quantum bounded storage model \cite[Theorem 6.4, Theorem 6.3]{Sch07}, there exists in Quamnesia a 1-out-of-2 oblivious transfer protocol, which has perfect security against malicious senders, and statistical security against malicious receivers.
Intuitively,  oblivious transfer is a primitive allowing a sender holding two messages, $m_0,m_1$, and a receiver holding a bit $\OTbit$, to communicate in such a way that the receiver learns $m_\OTbit$ but learns nothing about $m_{1-\OTbit}$. On the other hand, the sender learns nothing about $\OTbit$. 
As discussed in \cref{sec: MOE based security}, defining the security of oblivious transfer properly in an information-theoretic setting is a more contrived task than initially seems.
However, it is known that it suffices to construct a primitive called random oblivious transfer, where the sender has no input, but instead, random "inputs" are created during the protocol in such a manner that the receiver could only learn one of them. 
The Quamnesia random-OT scheme is described in \cref{alg: Amnesic OT}.

 It can easily be seen that the scheme is perfectly secure against malicious senders, as no data is sent from the receiver to the sender.
 Security against malicious receivers is more intricate and is detailed in \cite{Sch07}. 
We provide in \cref{sec: MOE based security} an alternative proof to that of \cite{Sch07} regarding security against malicious receivers based on the monogamy of entanglement property \cite{TFKW13}. For simplicity's sake, we prove a weaker security notion than that in \cite{Sch07}. However, by altering the proof's last stages, the same result as in \cite{Sch07} can be achieved (see \cref{rem: proof strengthening}).

\begin{algorithm*}
    \caption{Bit-Commitment Scheme $\AmCom$}
    \label{alg: Amnesic commitment}
\procedureblock {}{%
  \Radv(1^\secpar) \< \< \Cadv(1^\secpar,\CommBit) \pclb
 \pcintertext[dotted]{Commit Phase}\\
 \pcln \begin{minipage}[t]{12em}\raggedright
    $a \sample\Bitspace^\secpar$,$\theta\sample \Bitspace^\secpar$ ,\linebreak
      $\ket{\Stamp} =H^\theta\ket{a}$.
\end{minipage}\<\<\\
\pcln\< \sendmessageright*{\ket{\Stamp},\Stall} \<  \\   
\pcln\<\< \begin{minipage}[t]{10em}\raggedright
     Measure ${(H^\CommBit)}^{\tensor\secpar}\ket{\Stamp}$ \linebreak
     to obtain $\sigma$.
\end{minipage}\pclb
\pcintertext[dotted]{Reveal Phase}\\
\pcln \< \sendmessageleft*{\CommBit, \sigma} \< \\
\pcln \begin{minipage}[t]{12em}\raggedright
     If $\sigma$ is consistent\linebreak with $a$ on $\set{i|\theta_i=\CommBit}$,\linebreak
     output $\CommBit$, otherwise, abort.
\end{minipage}
\< \< }
\end{algorithm*}

\begin{algorithm*}
    \caption{Length $\ell$ Random Oblivious Transfer Scheme $\AmOT^\ell$ }
    \label{alg: Amnesic OT}
  
    \procedureblock {}{%
 \Sadv(\secparam)  \< \< \Radv(\secparam,\OTbit) \< \\
\pcln 
\begin{minipage}[t]{15em}\raggedright
\text{$\theta\sample \Bitspace^\secpar,x\sample\Bitspace^\secpar$},\linebreak 
$\ket{\Stamp}=H^\theta\ket{x}$.
\end{minipage}
 \<\<\< \\\label{stp: sadv 1}
\pcln 
 \<\sendmessageright*{\ket{\Stamp},\Stall} \< \< \\  
 \pcln\< \< \begin{minipage}[t]{15em}\raggedright
     Measure ${(H^\OTbit)}^{\tensor \secpar}\ket{\Stamp}$ \linebreak
     to obtain $\sigma$.
\end{minipage}   \< \\ \label{stp: radv 1}
\pcln 
\begin{minipage}[t]{15em}\raggedright
Let $\mathcal H$ be a universal hash family from $\Bitspace^{\leq \secpar}$ to  $\Bitspace^{\ell}$.
\begin{algorithmic}
\For{$c\in \Bitspace$}
\State  \text{Sample  a function $h_c\in \mathcal{H}$},
\State $\Cons^\theta_{c}\equiv\{i\in [\secpar]|\theta_i=c\}$,
\State $x_{c}\equiv x_{|{ \Cons^{\theta}_{c}}}$,
\State $m_{c}=h_c(x_{c})$.
\EndFor
\end{algorithmic}
\end{minipage}
 \<\<\< \\ \label{stp: sadv 2}
 \pcln \< \sendmessageright*{\langle h_0 \rangle ,\langle h_1 \rangle ,\theta} \<\< \\ 
 \pcln 
 \<\<  
\begin{minipage}[t]{15em}\raggedright 
\text{$\Cons^\theta_{\OTbit}\equiv\{i\in [\secpar]|\theta_i=\OTbit\}$}, \linebreak
$r\equiv h_\OTbit(\sigma_{|{ \Cons^{\theta}_{\OTbit}}})$.
\end{minipage}
 \< \\
 \label{stp: radv 2}\pcln \text{Output } (m_0,m_1).\<\< \text{Output  } r.
}
 We use $\langle h\rangle$ to refer to the description of a function $h\in \mathcal{H}$.
\end{algorithm*}

Both of these results are provably impossible to achieve in the standard quantum interactive model. A quantum bit-commitment that is both statistically hiding and statistically binding is known to be impossible in itself, even without perfect hiding. See \cite{May97, LC98} for impossibility results regarding commitments and \cite{GKR08, BGS13} for impossibility results on ideal oblivious transfer, also known as one-time memories.

By combining  the now classic reduction detailed in \cite{GMW86} with the existence of the commitment scheme $\AmCom$, the following result also holds in Quamnesia:
\begin{corollary}
In Quamnesia, there exist perfect quantum zero-knowledge proofs for $\NP$.
\end{corollary}

  The perfectness is due to the perfect hiding of $\AmCom$. Such a result is not known in the standard interactive classical or quantum models, even if statistical zero-knowledge proofs are considered instead. Interestingly, the corresponding complexity classes of decision problems for which there exists statistical zero-knowledge proofs possess some interesting qualities in the standard classic or quantum interactive models. Amongst those closure under complement (\cite{oka96, Wat02}) and the equivalence of honest and dishonest verification \cite{GSV98,wat09}, which additionally holds even in the computational setting (\cite{Vad06,Kob08}). We find it interesting to ponder whether analogs hold for the appropriate class in Quamnesia, either in computational, statistical, or perfect settings. A critical step appears to be finding a complete problem for the appropriate classes in Quamnesia.

Of similar interest is that the class of problems provable by computational zero-knowledge proofs contains $\IP$ if a non-interactive bit-commitment with statistical hiding exists \cite{BGGHKMR88}. 
While bit-commitment schemes exist in Quamnesia, the proof in \cite{BGGHKMR88} does not immediately transfer to Quamnesia. The proof has many stages, of which non are quite trivial. We highlight one particular challenge: the main idea of the proof is that in each round, the prover commits to messages of some interaction that would convince the verifier that she would accept. The prover then uses zero-knowledge for $\NP$ to convince the verifier that these commitments can indeed be opened to a valid interaction which results in acceptance. This approach does not work for the particular commitment presented in this paper. Since no data is sent from the committer to the receiver, the zero-knowledge proof only convinces the receiver of a trivial $\NP$ relation, which is the same for all messages.

Quamnesia additionally allows the construction of an almost ideal coin-flipping scheme. 
Originally introduced by \cite{Blu83}, a coin-flipping scheme allows Alice and Bob to jointly sample a common bit, which ideally should be sampled uniformly. At the end of the interaction, Alice and Bob respectively output a value  $\CoinBit_\adv, \CoinBit_\Badv \in \set{0, 1, \abort}$. If $\CoinBit_\adv=\CoinBit_\Badv$, then the protocol's output is $\CoinBit=\CoinBit_\adv$. If the values differ, then the protocol outputs $\CoinBit=\text{\abort}$. If Alice and Bob are both honest,  $\CoinBit$ distributes uniformly over $\Bitspace$. Let $P^*_\adv$ denote the supremum of $\max\{\Pr [\CoinBit = 0], \Pr [\CoinBit = 1]\}$ over all strategies of (a dishonest) Alice when Bob is honest, and let $\epsilon_\adv=P^*_{\adv}-\frac{1}{2}$ be Alice's Bias. Similarly, let $P^*_\Badv$ denote the supremum of $\max\{\Pr [\CoinBit = 0], \Pr [\CoinBit = 1]\}$ over all strategies of (a dishonest) Bob when Alice is honest, and let $\epsilon_\Badv=P^*_{\Badv}-\frac{1}{2}$ be Bob's Bias.

\begin{corollary}
In Quamnesia, there exist a strong \footnote{ Strong coin-flipping contrasts with weak coin-flipping by the goals of the parties: in weak coin-flipping, Alice hopes to flip 0, and Bob hopes to flip 1; hence Bob being able to force the result 0 with probability $\frac{3}{4}$ would only be detrimental to himself, and thus would not be considered cheating; in strong coin-flipping, any bias may be beneficial to a cheating party, and hence any non-uniform distribution of the coin is considered cheating.} coin-flipping scheme  \cite{Kit02,CK09}, detailed in \cref{alg: amnesic coin-flipping}, which is one-sided perfect, and has arbitrarily small bias for the other side. 
\end{corollary}
The scheme itself, detailed in \cref{alg: amnesic coin-flipping}, is the result of a well known construction of coin-flipping via commitments~\cite{BC90}. When instantiated with \cref{alg: Amnesic commitment}, it results in a coin-flipping scheme where  $\epsilon_\Badv=0$ and $\epsilon_\adv$ is negligibly small in the security parameter of $\AmCom$.
Negligibly small biases are known to be impossible even in the quantum standard interactive model \cite{LC98,Kit02} and are in clear contrast to a result by Kitaev:
\begin{theorem}[{\cite{Kit02}}] 
\label{thm: coin-flipping bounds}
 In any coin-flipping scheme, $P^*_\adv P^*_{\Badv}\geq \frac{1}{2}$. In particular, either $\epsilon_\adv\geq \frac{1}{\sqrt{2}}-\frac{1}{2}$, or $\epsilon_\Badv\geq \frac{1}{\sqrt{2}}-\frac{1}{2}$; and if $\epsilon_\Badv=0$, then  $\epsilon_\adv=\frac{1}{2}$. 
\end{theorem} 
\begin{remark}
Note that the last statement proves that it is impossible to have a non-trivial one-sided perfect security coin-flipping in the standard quantum interactive model.
\end{remark}

\begin{algorithm*}
    \caption{Ideal Coin-Flipping $\Amflip$ , based on the bit-commitment scheme $\AmCom$ }
\label{alg: amnesic coin-flipping}

\procedureblock { }{%
  \adv(1^\secpar) \< \< \Badv(1^\secpar) \< \\
 \pcln a\sample \Bitspace.  \\
\pcln\< \sendmessagerightleft*{
\begin{aligned}
 &\text{Using $\AmCom$, $\adv$ interacts}\\ &\text{with $\Badv$ to commit to $a$}
\end{aligned}
} \< \< \\   
 \pcln\< \<  b\sample \Bitspace \< \\
\pcln \< \sendmessageleft*{b} \<\< \\
\pcln\< \sendmessagerightleft*{
\begin{aligned}
&\text{Using $\AmCom$, $\adv$ interacts}\\ &\text{with $\Badv$ to reveal $a$}
\end{aligned}} \< \< \\   
\pcln \< \< 
\begin{minipage}[t]{20em}\raggedright
    Let the outcome of the reveal be $a'$.\linebreak
    If the reveal failed, $\CoinBit_{\Badv}=\abort$,\linebreak
    otherwise,  $\CoinBit_{\Badv}=a' \oplus b$.
\end{minipage}
\<\\
\pcln  \text{ Output $\CoinBit_{\adv}=a \oplus b$.}\< \< \text{ Output $\CoinBit_{\Badv}$,}}
\end{algorithm*}

\paragraph{Skepticland vs. Quamnesia}
\setlength{\epigraphwidth}{0.9\textwidth}
\epigraph{All quantum optimists are alike; each quantum skeptic is skeptic in their own way.}{\textit{---Gil Kalai} (private communication), paraphrasing the first line of Tolstoy's  Anna Karenina.}

As Kalai's quote above states, it is notoriously hard to characterize quantum skeptics. In this section, we discuss the relationship between Quamnesians---those that believe that quantum computers cannot store quantum states for long periods---and quantum skeptics. We argue that Quamnesians constitute a dominant faction of quantum skepticism; we also present notable other arguments in the other faction, that is,  other quantum skeptic arguments unrelated to Quamnesia.

The existence of efficient universal quantum error-correction for large-scale systems, and by extension, long-term quantum memory, is considered vital to quantum computation by quantum optimists and skeptics alike \cite[Chapter 1]{Dya20}.
It is generally agreed that quantum long-term memory requires  efficient universal quantum error-correction.
As a consequence, we believe that common skeptic views such as the following do support the notion that quantum computers will forever remain "amnesic":
\begin{displayquote}[{\cite[Section 1]{Kal20}}]
\textit{Noisy quantum systems will not allow building quantum error-correcting codes needed for quantum computation.}
\end{displayquote}

We argue Dyakonov---another prominent quantum skeptic---is a Quamnesian: he claims~\cite[Chapter 4]{Dya20} that the assumptions in the threshold theorem~\cite{AB08}, which are also invaluable to quantum error-correction, are not well justified.

 However, we emphasize that not all quantum skeptics necessarily believe we live in Quamnesia: there are other objections to quantum supremacy. Some skeptics may claim that only general-purpose fault-tolerant quantum computation is impossible, meaning quantum long-term memory can be stored, but it will not be possible to perform universal computation on it without significant errors. We note though that if quantum memory is available, provided that the technology of storage still allows the relatively simple measurements and unitaries that can be performed today, then significant advancements in quantum cryptography are implied still; even being able to reliably store, prepare and perform simple measurements on the simplest of states already implies the practical application of more quantum cryptographic primitives that do not exist in a classical world. Amongst those primitives are private quantum money~\cite{Wie83, PYJ+12, MVW12}, tokens for $\mac$~\cite{BSS21}, uncloneable encryption~\cite{BL20, AK21}, and quantum encryption with certified deletion~\cite{BI20}, as they all have constructions that are realistic except for the need to store BB84 states.
Perhaps, one could contemplate a device that can reliably store simple states, such as the BB84 states, but \emph{not} more general (e.g., entangled) states. However, this is impossible: if a quantum device preserves BB84 states (or any other orthogonal basis for the vector space of $2^n\times 2^n$ density matrices) without significant errors, then by linearity, that device can be represented as a quantum channel that is close to the identity, and hence would preserve any state whatsoever. Therefore, as long as there is a way to load entangled states to the device, these states would be stored as reliably as the BB84 states.
A different form of skepticism is \cite{Ber18}. Bernstein does not make any speculation about future technology but instead attacks the security claims of current quantum protocols. Bernstein claims that QKD could be \emph{incompatible}
with the laws of physics. One of his main arguments is that by the holographic principle, there is no way for Alice (or Bob) to keep any secrets from Eve. 
Other motivations for quantum skepticism stem from a more practical real world perspective, for example \cite{Var19b}, which argues that quantum computing would not be realized in practice because the small number of applications lacks the "positive feedback loop" of successful algorithms which has driven classical computing forwards. Aaronson \cite{Aar13b}, too, lists some common arguments of  skeptics---along with his counterarguments.

%%%%%%%%%%%%%%%%%%%%%%%%%%%%
\ifnum\masterthesis=0
    \subsection*{Acknowledgments}
\fi
%%%%%%%%%%%%%%%%%%%%%%%%%%%%
%NON-ANON PART

\ifnum\anon=0
%ANON PART
We wish to thank Amos Beimel, Anne Broadbent, Amit Behera, and Gil Kalai for valuable discussions. This work was supported by the Israeli Science Foundation (ISF) grant No. 682/18 and 2137/19, and by the Cyber Security Research Center at Ben-Gurion University.

\BeforeBeginEnvironment{wrapfigure}{\setlength{\intextsep}{0pt}}
\begin{wrapfigure}{r}{90px}
    %\centering
    \includegraphics[width=40px]{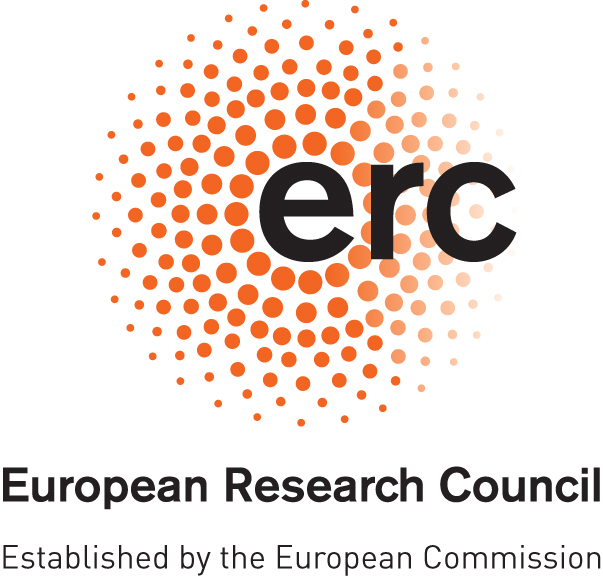}  \includegraphics[width=40px]{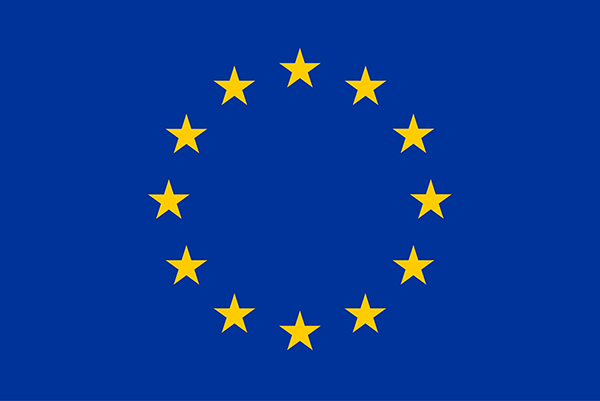}
\end{wrapfigure}This work has received funding from the European Research Council (ERC) under the European Union’s Horizon 2020 research and innovation programme (grant agreement No 756482).

\fi

%%%%%%%%%%%%%% HANDLING OF BIBLIOGRAPHY
%\nocite{*} %to check if the added bib entries compile correctly with the old file.
\ifnum\sigconf=1
    \bibliographystyle{ACM-Reference-Format}
\else
    \ifnum\cryptology=1
        \bibliographystyle{abbrv}
    \else
        \bibliographystyle{alphaabbrurldoieprint}

    \fi
\fi

\ifnum\masterthesis=0
    \ifnum\smallbib=1
        {\footnotesize \bibliography{main} }
    \else 
        \bibliography{main}
    \fi
\fi

\appendix
\ifnum\shownomenclature=1
\printnomenclature[1in]
%There is a label which is created automatically in the macros of the following form: \label{sec:nomenclature}
\fi
\ifnum\masterthesis=1
    \ifnum\smallbib=1
        {\footnotesize \bibliography{main} }
    \else 
        \bibliography{main}
    \fi
\fi

\section{Security of Oblivious Transfer in Quamnesia} \label{sec: MOE based security}
Oblivious transfer (\OT) is a primitive allowing a sender holding two uniformly random messages of length $\ell$,  $m_0,m_1$, and a receiver holding a bit $\OTbit$ to communicate in such a way that the receiver learns $m_\OTbit$, but learns nothing of $m_{1-\OTbit}$, and on the other hand, the sender learns nothing of $\OTbit$. A different primitive, called   Random$\textit{-}\OT$, implies $\OT$. In Random$\textit{-}\OT$, the "inputs" of the sender are created "on the fly," and yet the receiver still can know one---and only one--of the two "inputs," while the sender does not know which of the two the receiver knows. To get a true $\OT$, the sender can perform the random $\OT$ protocol with the receiver, xor his true inputs with the "on the fly inputs," and send the result to the receiver. The receiver is thus only able to uncover one of the two inputs. The actual security definitions for $\OT$ are more delicate than may be expected. The reader is hence referred to \cite{CSSW06} for a more complete discussion on  the security definitions of $\OT$, and the reduction to Random$\textit{-}\OT$. 
\begin{definition}[Random-$\OT^\ell$, based on \cite{Sch07}]
Random-$\OT^\ell$ is an interactive protocol between  a receiver $\Radv$ and a sender $\Sadv$ receiving input $1^\secpar$. $\Radv$ additionally has an input $b\in\Bitspace$. A random-$\OT^\ell$ protocol is complete if, at the end of communication between the honest parties, neither of the parties aborts, $\Sadv$ has output $m_0$, $m_1\in \Bitspace^\ell$, which distribute uniformly at random, and $\Radv$ has
output $r=m_\OTbit$.
\end{definition}

The completeness of $\AmOT^\ell$ is immediate.
\begin{theorem}\label{thm: AMOT Completeness}
$\AmOT^\ell$, specified in \cref{alg: Amnesic OT}, is complete.
\end{theorem}

Recall that~\cite{Sch07} proved that $\AmOT^\ell$ is also secure. In this section, we use the monogamy of entanglement property \cite{TFKW13} to provide an alternative proof for its security. The alternative proof further demonstrates the usefulness of the monogamy of entanglement tool and results in a proof that is arguably more compact. Furthermore, it suggests that construction based on coset states \cite{CLLZ21} instead of BB84 states is also secure due to an analogous property for these states \cite{CV21}.  In other cryptographic primitives, constructions based on coset states or their simplified variant, subspace states \cite{AC13}, often possess valuable properties that constructions based on BB84 states do not (compare, for example, \cite{AC13, BS16} with the weaker primitives in \cite{MVW12, BSS21}).
For simplicity, we only prove a weaker security notion than that described in \cite{Sch07}. As stated in \cref{rem: proof strengthening}, the latter part of the proof can easily be altered to achieve the stronger security notion.

\begin{definition}[Receiver and Sender Guess-Resistant Security]\label{def: guess-resistant security}
 For a given Random-$\OT^\ell$ scheme $\Pi$ with an honest sender $\Sadv$ and an honest receiver $\Radv$:
 \begin{itemize}
\item The sender's advantage is  $\epsilon_\Sadv:= 
\Sup_{\Sadv^*}{\Pr[\Sadv^* \text{ guesses \ensuremath{\OTbit} and \ensuremath{\Radv} does not Abort}]}-\frac{1}{2}$. We say that $\Pi$ is receiver secure if $\epsilon_\Sadv$ is negligible in $\secpar$.
\item The receiver's advantage  is $\epsilon_\Radv:= \Sup_{\Radv^*}{\Pr[\Radv^*\text{ guesses \ensuremath{(m_0, m_1)} and \ensuremath{\Sadv} does not Abort}]}-\frac{1}{2^\ell}$.   We say that $\Pi$ is sender secure if $\epsilon_\Radv$ is negligible in $\secpar$.
\end{itemize}
\end{definition}

\begin{remark} \label{rem: proof strengthening}
The same security notion as \cite{Sch07} can be achieved
by following the proof in \cite{Sch07}  after \cref{lem: joint Min entropy bound }. 
\end{remark}

\begin{theorem}\label{thm: AmOT guess resistant security}
In Quamnesia, $\AmOT^\ell$  (\cref{alg: Amnesic OT}) satisfies receiver and sender guess-resistant security for any $\ell\in \mathbb{N}$.
\end{theorem}

\begin{proof}
 We first argue for security against a malicious sender. In $\AmOT$ (see \cref{alg: Amnesic OT}), no data is transmitted from the receiver to the sender; hence the view of any given $\Sadv^*$ is the same regardless of whether $\OTbit=0$ or $\OTbit=1$, meaning $\epsilon_{\Sadv}=0$.

 Security in light of a malicious receiver is more complex. Let $\Radv^*$ be an unbounded malicious receiver (but subject to the limitations of Quamnesia). It will be convenient to split the action of  $\Radv^*$ into two: $\Radv^*_1$, describing the action of $\Radv^*$  in place of step~\ref{stp: radv 1} of the protocol, which receives a quantum state, and keeps only some classical memento of it $w$;  and $\Radv^*_2$, describing the action of $\Radv^*$ in place of step~\ref{stp: radv 2} of the protocol,  receiving the transmission $\theta, h_0,h_1$ in addition to the classical memento $w$, and outputting guesses $(m'_0,m'_1)$. 
 
 \begin{definition} \label{def: basic variables}
For any given $\Radv^*_1$, let  $\WRadv$ denote the random variable describing the distribution of the string $w$ representing the classical state kept by $\Radv^*_1$ following the \Stall.  Let $\Theta$ be the random variable describing the choice of $\theta$ in step~\ref{stp: sadv 1} of the protocol. Similarly, let  $X_0, X_1$ be the random variables describing the choices of $x_0,x_1$ in step~\ref{stp: sadv 2}   of the protocol. \end{definition}

For the first part of the proof, we will see that for any  $\WRadv$ produced by $\Radv^*_1$, it is difficult to guess both $X_0$ and $X_1$ simultaneously.
\begin{lemma} \label{lem: joint Min entropy bound }
For any $\Radv_1^*$, $H_\infty(X_0,X_1|\WRadv,\Theta)\geq \secpar\log(\MOEExpInv)$.
\end{lemma}
 The lemma is an application of a result by \cite{TFKW13}, analyzing the following game between three players, Alice, Bob and Charlie:
 \begin{savenotes}
\begin{game}
 \caption{  Monogamy Of Entanglement Game $\MOE(\secpar)$
} 
\label{exp:MOE} 
\begin{algorithmic}[1] 
    \State Bob and Charlie create a tripartite quantum state $\rho_{ABC}$ and split it to three. Alice receives the $\secpar$ qubits $A$ subsystem, Bob receives the $B$ subsystem, and Charlie receives the $C$ subsystem. From this point onwards, Bob and Charlie do not communicate anymore.  
    \State \label{stp: Alice }  Alice samples  $\theta\sample \Bitspace^\secpar $ uniformly, and performs the measurement $\set{\ketbra{x_\theta}}_{x\in \Bitspace^\secpar}$\footnote{ $\ket{x_{\theta}}=H^{\tensor\theta}\ket{x}$. } to receive some string $x$. 
    \State Alice sends $\theta$ to both Bob and Charlie.
    
    \State Bob outputs a string $y$, and Charlie outputs a string $z$.
\end{algorithmic}
The value of the game is $1$ if $x=y=z$. 
 \end{game}
\end{savenotes}

 \begin{theorem}[{Monogamy Of Entanglement, adapted from~\cite[Theorem 3.4]{TFKW13}}] \label{thm: MOE bound}
For any (computationally unbounded) Bob and Charlie:

\[\Pr[\MOE(\secpar)=1]\leq \MOEExp^\secpar.\]
 \end{theorem}

 \begin{proof}[Proof of \cref{lem: joint Min entropy bound }]
 Let $\Radv^*$ be an unbounded malicious receiver such that \[H_\infty(X_0,X_1|\WRadv,\Theta)=\log\Big{(}\frac{1}{q(\secpar)}\Big{)},\]
for an appropriate function $1\geq q(\secpar)>0$. Recall the operational meaning of min-entropy as the guessing probability: when receiving a sample from $\WRadv$ and $\Theta$, there exists a (not necessarily efficient) algorithm $\Fadv$ that can correctly guess $(X_0,X_1)$ with probability $q(\secpar)$. By standard convexity arguments, it is easy to argue $\Fadv$ is W.L.O.G deterministic. Let $\mathcal{C}$ be the channel which is the same as the one $\Radv^*_1$ applies but copies (with CNOT) the result in two identical registers $B, C$. Observe the following strategy by Bob and Charlie in Game~\ref{exp:MOE}:
\begin{enumerate}

    \item Prepare the maximally entangled state between $A$ and $D$ (where $D$, like $A$, is a $\secpar$ sized quantum register.)   $\ket{\phi}=\frac{1}{\sqrt{2^{2\secpar}}}\sum_{x\in \Bitspace^{ \secpar}}\ket{x}_A \ket{x}_D$.
    \item Apply $\mathcal{C}$ on the $D$ subsystem. Bob receives the subsystem on $B$, and Charlie the one on $C$.
    \item Alice Performs a measurement as detailed in  step~\ref{stp: Alice } of Game~\ref{exp:MOE}, and sends $\theta$ to Bob and Charlie.
    \item Following Alice's measurement, Bob and Charlie's registers both contain an identical classical string $w$---as specified by the action of $\Cadv$. Both now independently perform the optimal algorithm $\Fadv$ on $w$. Hence both derive an identical pair of strings $(y_0,y_1)=(z_0,z_1)$. Bob and Charlie then reorder their strings according to $\theta$ to get identical outputs $y=z$.
\end{enumerate}

Observe that for any $\theta$ \[\ket{\phi}=\frac{1}{\sqrt{2^{2\secpar}}}\sum_{x\in \Bitspace^{ \secpar}}\ket{x}_A \ket{x}_D=\frac{1}{\sqrt{2^{2\secpar}}}\sum_{x\in \Bitspace^{ \secpar}}\ket{x_\theta}_A \ket{x_\theta}_D.\]
This fact is evident by noticing that the inner product of the two expressions is $1$.  Ergo, after the action of $\Cadv$, the tripartite state can be written as $\frac{1}{2^{2 \secpar}}\sum_{x,x'\in \Bitspace^{ \secpar}}\ket{x_\theta}\bra{x'_\theta}_A\tensor \mathcal{C}(\ket{x_\theta}\bra{x'_\theta})_{B,C}$, and following Alice's measurement, it is $\ket{x_\theta}\bra{x_\theta}_A\tensor \mathcal{C}(\ket{x_\theta}\bra{x_\theta})_{B, C}$, for  $x,\theta$ that are according to the sampling and the result of the measurement done by Alice.

The classical results stored in both $B$ and $C$  after applying $\mathcal{C}$ are identical and are distributed the same as in \cref{alg: Amnesic OT}.
Since the distribution of $x,\theta$ is also the same as in \cref{alg: Amnesic OT}, and because $\Fadv$ is deterministic, we know that with probability $q(\secpar)$ Bob and Charlie are both correct with their guesses for $x_0,x_1$, i.e., $y_0=z_0=x_0$, and $y_1=z_1=x_1$. In that event, $y=z=x$ and thus $q(\secpar)\leq \MOEExp^{\secpar}$.
As a result of the bound in \cref{thm: MOE bound}, \cref{lem: joint Min entropy bound } is immediately implied by observing that $\MOEExp^{-1}=\MOEExpInv$.
 \end{proof}
 Our next step demands some caution: we know   $X_0$ and $X_1$ must be hard to guess simultaneously given $\WRadv$ and $\theta$; so it appears plausible that one of the two has high min-entropy conditioned on $\WRadv$ and $\theta$. However, this is false: it is not hard to come up with $\WRadv$ that allows one strategy to guess  $X_0$ with probability $\frac{1}{2}$ and another strategy  to guess $X_1$ with probability $\frac{1}{2}$ while no strategy can guess both simultaneously with probability more than $\MOEExp^{\secpar}$ (as \cref{lem: joint Min entropy bound } dictates). In order to overcome this difficulty, a weaker statement must be made instead: first, we argue not over the entropy of $X_0$ or $X_1$  but of the entropy of some "mix" of the two, and second, due to the pathologies of the min-entropy, we are forced to argue not that the min-entropy of that mix is large, but that its smooth min-entropy \cite{Ren08,RW05}---for some $\delta>0$---is large. 
  \begin{lemma}
  \label{lem: Smooth min entropy bound}
 For any $\Radv_1^*$  there exists a binary random variable $\CRadv$ such that for any $\delta>0$, \[H^{\delta}_\infty(X_{\CRadv}|\WRadv,\Theta,\CRadv)\geq\frac{\secpar}{2}\log(\MOEExpInv)-1-\log(\frac{1}{\delta}).\]
  \end{lemma}
  \cref{lem: Smooth min entropy bound} is derived from \cref{lem: joint Min entropy bound } via the following corollary by \cite{Sch07}.
  \begin{corollary}
[{Corollary of Min-Entropy Splitting Lemma	, \cite[ Corollary 2.16]{Sch07}}]
 \label{cor: split lemma cor}
Let $X_0$, $X_1$, $Z$ be random variables with
$H_{\infty}(X_0,X_1|Z) \geq \alpha$. Then, there exists a binary random variable $C$ such
that for every $\delta>0$
$H^{\delta}_\infty(X_{1-C}|(Z,C))\geq \frac{\alpha}{2}-1-\log(\frac{1}{\delta})$.
  \end{corollary} 
   \cref{cor: split lemma cor} is itself achieved by combining the min-entropy splitting lemma,  first proved in a preliminary version of \cite{Wul07}, with the min-entropy chain rule \cite{Ren08}. We will not repeat the details here.
   
 For the last part of the proof, we recall a smooth min-entropy variant of the leftover hash lemma and use it to provide a bound on the probability of any $\Radv^*$ to guess $m_0$ and $m_1$ simultaneously.
 \begin{lemma}[{Leftover Hash Lemma for Smooth Entropy With Side Information, based on \cite[Corollary 5.6.1.]{Ren08}}]
 Let $\Hash$ be a family of universal hash functions from $D$ to $\Bitspace^\ell$, i.e., for any $x \neq  x'$,
$\Pr_{h\sample H} [h(x) = h(x')] = 2^{-\ell}$.  Let $(X,Y)$ be a joint random variable where $X$ is a classical random variable over $D$, and $Y$ is a classical random variable. Then, for any $\delta\geq 0$,
\[||  (\langle h \rangle,h(X),Y)-(\langle h \rangle,U_\ell ,Y)||_1\leq  2^{-{\frac{1}{2}(H^{\delta}_\infty(X|Y)-\ell)}}+2\delta,\] 
where in the above equation, $h$ is sampled uniformly at random from $\Hash$, and $U_\ell$ distributes uniformly on $\ell$ bits strings. 
\end{lemma}
  Plugging  \cref{lem: Smooth min entropy bound} with $\delta=2^{-\frac{\secpar}{20}},X= X_{\CRadv}$, and $Y= (\WRadv,\Theta,\CRadv)$ to the Leftover Hash Lemma, we get that for any $\secpar$:

\begin{equation}\label{eq: hash lemma bound 2}
    || (\langle h \rangle , h(X_{\CRadv}),\WRadv,\Theta,\CRadv)-(\langle h \rangle , U_\ell,\WRadv,\Theta,\CRadv)||_1\leq {2^{\frac{1}{2}(\ell-{\frac{ \secpar}{2}\log(\MOEExpInv)+1+\frac{\secpar}{20})}}}+2^{-\frac{\secpar}{20}+1}.
\end{equation}

Observe that
\begin{align*}
 &2^{\frac{1}{2}(\ell-{\frac{\secpar}{2}\log(\MOEExpInv)+1+\frac{\secpar}{20})}}=2^{\frac{1}{2}(\ell+1+\frac{\secpar}{20})-{\frac{\secpar}{4}\log{(\MOEExpInv)}}}\\
 \leq&2^{\ell+1+\frac{\secpar}{20}-{\frac{\secpar}{4}\log{(\MOEExpInv)}}} \leq2^{\ell+1+\frac{\secpar}{20}-0.0571\secpar}\\
=&2^{\ell+1-0.0071\secpar},\end{align*}

meaning that for any $\secpar$
\begin{equation}\label{eq: hash lemma bound 3}
    || (\langle h \rangle , h(X_{\CRadv}),\WRadv,\Theta,\CRadv)-(\langle h \rangle , U_\ell,\WRadv,\Theta,\CRadv)||_1\leq 2(2^{\ell-0.0071\secpar}+2^{-\frac{\secpar}{20}}).
\end{equation}

To finish the proof, we present the following lemma:
\begin{lemma}\label{lem:distinguisher between uniform and hash}
 Let $\Radv^*=(\Radv^*_1,\Radv^*_2)$  be a malicious receiver guessing both $m_0$ and $m_1$ correctly with probability $\frac{1}{2^\ell}+\epsilon(\secpar)$ without the sender aborting. Then there exists a distinguisher $\Dadv$ distinguishing between the two distributions described in the l.h.s. of \cref{eq: hash lemma bound 3} with an advantage at least $\epsilon(\secpar)$.
\end{lemma}
 
The statistical distance in \cref{eq: hash lemma bound 3} denotes the optimal advantage of any distinguisher in distinguishing between the two distributions. Hence from \cref{lem:distinguisher between uniform and hash} we get that
\[\epsilon(\secpar)\leq 2(2^{\ell-0.0071\secpar}+2^{-\frac{\secpar}{20}}),\]
meaning $\epsilon(\secpar)$ decreases exponentially in $\secpar$. In turn, $\epsilon_{\Radv}$, the supremum over all $\Radv^*$, is negligible in $\secpar$. This concludes the proof of \cref{thm: AmOT guess resistant security}.
\end{proof}

\begin{proof}[Proof of \cref{lem:distinguisher between uniform and hash}]
$\Dadv$, on input $(\langle h \rangle , \tilde{m}, w,\theta,c)$, will simulate $\Radv^*_2$ (the action of $\Radv^*$ in step~\ref{stp: radv 2}) on inputs $\langle \tilde{h}_{0} \rangle, \langle \tilde{h}_{1} \rangle ,\theta$ in addition to  the kept state $w$, where $\tilde{h}_{c}=h$, and $\tilde{h}_{1-c}$ is sampled from $\Hash$. $\Radv^*_2$ responds with guesses $m'_0, m'_1$ and $\Dadv$ responds with $1$ if $\tilde{m}=m'_{c}$, and $0$ otherwise. Notice that regardless of the value of $c$, $\langle \tilde{h}_{0} \rangle, \langle \tilde{h}_{1} \rangle$ distribute uniformly at random and independently of each other, same as $\langle {h}_{0} \rangle, \langle {h}_{1} \rangle$ do in true communication of $\Radv^*$ with the sender. 
 If $\tilde{m}$ is a random string, it is clear that $\Dadv$ will respond $1$ with probability exactly $\frac{1}{2^\ell}$. If, however, $\tilde{m}=h(x_c)$, we observe that $m'_0,m'_1$ will distribute the same as in true communication of $\Radv^*$ with the sender, and hence $\tilde{m}=m'_{c}$ with probability at least $\frac{1}{2^\ell}+\epsilon(\secpar)$, implying $\Dadv$'s advantage is at least $\epsilon(\secpar)$. 
\end{proof}

\end{document}